\documentclass[11pt]{article}
\usepackage{verbatim,amssymb,amsthm,amsfonts}
\usepackage[nonamelimits]{amsmath}
\usepackage{fullpage}

\providecommand\floor[1]{\lfloor#1\rfloor}
\providecommand\ceil[1]{\lceil#1\rceil}

\providecommand\abs[1]{\lvert#1\rvert}

\def\poly{\mathrm{poly}}

\newtheorem{theorem}{Theorem}

\newtheorem{claim}[theorem]{Claim}
\newtheorem{conjecture}[theorem]{Conjecture}

\DeclareMathOperator\pr{\mathrm{Pr}}
\DeclareMathOperator\E{\mathrm{E}}

\def\eps{\varepsilon}

\def\calC{\mathcal{C}}
\def\Sbar{\overline{S}}
\def\Tbar{\overline{T}}

\begin{document}
\title{A better tester for bipartiteness?}
\author{Andrej Bogdanov\footnote{{\tt andrejb@cse.cuhk.edu.hk}. Department of Computer Science and Engineering and Institute for Theoretical Computer Science and Communications, Chinese University of Hong Kong.}  \and Fan Li\footnote{{\tt fli@cse.cuhk.edu.hk}. Department of Computer Science and Engineering, Chinese University of Hong Kong.}}
\date{}

\maketitle

\begin{abstract}
Alon and Krivelevich (SIAM J. Discrete Math. 15(2): 211-227 (2002)) show that if a graph is $\eps$-far from bipartite, then the subgraph induced by a random subset of $\tilde{O}(1/\eps)$ vertices is bipartite with high probability. We conjecture that the induced subgraph is {\em $\tilde{\Omega}(\eps)$-far} from bipartite with high probability. Gonen and Ron (RANDOM 2007) proved this conjecture in the case when the degrees of all vertices are at most $O(\eps n)$. We give a more general proof that works for any $d$-regular (or almost $d$-regular) graph for arbitrary degree $d$.

Assuming this conjecture, we prove that bipartiteness is testable with one-sided error in time $O(1/\eps^c)$, where $c$ is a constant strictly smaller than two, improving upon the tester of Alon and Krivelevich. As it is known that {\em non-adaptive} testers for bipartiteness require $\Omega(1/\eps^2)$ queries (Bogdanov and Trevisan, CCC 2004), our result shows, assuming the conjecture, that adaptivity helps in testing bipartiteness.
\end{abstract}

\section{Introduction}

A graph on $n$ vertices is {\em $\eps$-far from bipartite} if it cannot be made bipartite even after removing an arbitrary set of $\eps\binom{n}{2}$ of its edges. A randomized algorithm is a {\em (one-sided) tester for bipartiteness} with distance parameter $\eps$ if, given oracle access to a graph $G$, the algorithm always accepts if $G$ is bipartite and rejects with probability at least $1/2$ if $G$ is $\eps$-far from bipartite.

The problem of testing bipartiteness was among the first ones considered in the framework of property testing. In their work which introduces combinatorial property testing, Goldreich, Goldwasser and Ron~\cite{GGR} give a tester for bipartiteness (of dense graphs) that runs in time $O(1/\eps^3)$. 

Later works~\cite{AFKS,AS} gave many examples of graph properties that are testable in time independent of the size of the graph. Alon et al.~\cite{AFNS} gave a combinatorial characterization of all such graph properties. However, although all these properties have testers whose running time is independent of the graph size, the running time is very large in terms of the distance parameter $\eps$. Typically this dependence is a tower of exponentials of height polynomial in $1/\eps$, and sometimes even worse. The super-exponential dependence on $\eps$ in these analyses owes to applications of the regularity lemma. In contrast, there are relatively few properties that are testable in time polynomial in $1/\eps$. Bipartiteness is perhaps the most natural example of such a property.

In most cases, the design of property testing algorithms for graphs is relatively straightforward; it is the analysis that is difficult. This can be partially explained by a theorem of Goldreich and Trevisan~\cite{GT} which says that any tester that makes at most $q$ queries can be simulated by a tester that chooses a subgraph on $O(q)$ vertices, makes all queries between the vertices and accepts if the induced subgraph satisfies a related property. Applying this canonical algorithm of Goldreich and Trevisan deteriorates the query complexity by at most a quadratic factor. In the case of superexponential bounds (typical in applications of the regularity lemma), a quadratic loss is not terribly relevant.\footnote{The deterioration in {\em running time} could be somewhat worse, but again such losses are not very relevant when using the regularity lemma.}

However, for bipartiteness testing, much sharper bounds on the query complexity are known. Alon and Krivelevich~\cite{AK} show that if a graph is $\eps$-far from bipartite, then with high probability a random sample of $\tilde{O}(1/\eps)$ vertices contains an odd-length cycle. This gives a bipartiteness tester that makes $\tilde{O}(1/\eps^2)$ queries, improving the query complexity of the original tester of Goldreich, Goldwasser and Ron by a factor of about $1/\eps$.

The Alon-Krivelevich tester is {\em non-adaptive} in the sense that its query sequence is independent of the graph in question (and only depends on the parameter $\eps$). Bogdanov and Trevisan~\cite{BT} show that any non-adaptive tester for bipartiteness must make $\Omega(1/\eps^2)$ queries into the graph. For general testers, they only prove a query complexity lower bound of $\Omega(1/\eps^{3/2})$.

It is natural to ask whether the Alon-Krivelevich tester can be improved. Gonen and Ron~\cite{GR} observe that the known lower bounds apply even for the subclass of graphs where the degree of every vertex is bounded by $O(\eps n)$. For this special class of graphs, they give an improved algorithm for testing bipartiteness that runs in time $\tilde{O}(1/\eps^{3/2})$. 

Gonen and Ron also give a different type of algorithm that tests bipartiteness for any graph where all vertices have degree {\em at least} $k\eps n$ ($1 \leq k \leq 1/\eps$) with running $\tilde{O}(1/k\eps^2)$. Although not stated explicitly in their paper, when taken together their results imply that bipartiteness is testable in time $O(1/\eps^{2-c})$ for some constant $c < 2$ for any $d$-regular graph, where $d$ can be arbitrary. (We explain this point below. In fact it is sufficient for the graph to be ``almost $d$-regular'' in the sense that the degrees of most vertices are between $d$ and $Kd$ for some constant $K$.)

Recently Goldreich and Ron~\cite{GolR09} gave the first examples of graph properties for which an adaptive tester provably outperforms the best non-adaptive one by a polynomial factor.\footnote{Although Gonen and Ron already show that adaptivity helps for testing bipartiteness of small-degree graphs, their result does not translate into a graph property because their notion of ``small degree'' depends on the distance parameter $\eps$.} They also give a candidate collection of properties which is conjectured to approach the optimal quadratic speedup for adaptive testers.

The question of testing bipartiteness has also been studied in the {\em sparse graph} model, where Goldreich and Ron~\cite{GolR98} give a tester whose running time on an $n$-vertex graph is $\tilde{O}(\sqrt{n})$ for any fixed degree $d$ and distance parameter $\eps$. Although their discussion focuses on the case of {\em constant} $d$ and $\eps$, their analysis in fact works even when $d$ and $1/\eps$ grow in terms of $n$ (although the bounds get worse). For a precise statement of this fact, see Theorem~\ref{thm:gr}.

Kaufman, Krivelevich, and Ron~\cite{KKR} gave upper and lower bounds for testing bipartiteness for general-degree graphs in a model that combines features of the sparse graph and dense graph models.

\subsection{Our result}

Recall that Alon and Krivelevich~\cite{AK} prove that if a graph is $\eps$-far from bipartite, then a random sample of $\tilde{\Omega}(1/\eps)$ vertices will contain an odd-length cycle with probability at least $1/2$. We conjecture the following generalization:

\begin{conjecture}
\label{conj:ak}
If a graph $G$ is $\eps$-far from bipartite, then with probability $1/2$ the induced graph on a random sample of $\tilde{O}(1/\eps)$ vertices is $\tilde{\Omega}(\eps)$-far from bipartite.
\end{conjecture}

Gonen and Ron~\cite{GR} prove this conjecture for graphs whose degrees are bounded by $K\eps n$ for some constant $K$. We give a more general proof (Theorem~\ref{thm:ak}) that works for the case where $G$ is $d$-regular (and even when all its vertex degrees are between $d$ and $Kd$ for a constant $K$). 

Conjecture~\ref{conj:ak} is a purely combinatorial statement about graphs, which by itself does not yield anything better than the original, nonadaptive tester of Alon and Krivelevich for testing bipartiteness in time $\tilde{O}(1/\eps^2)$. Assuming this conjecture, we prove that there exists an {\em adaptive} tester for bipartiteness that outperforms the best non-adaptive one:

\begin{theorem}
\label{thm:main}
Assume Conjecture~\ref{conj:ak}. Then bipartiteness is testable with one-sided error in time $O(1/\eps^c)$, where $c$ is a constant strictly less than 2.
\end{theorem}

\paragraph{On the conjecture} Besides our proof that the conjecture is true for regular graphs, more evidence in favor of it comes from work of Fischer and Newman \cite{FN}. It follows from their work that there exists a function $f$ such that if $G$ is $\eps$-far from bipartite, thaen with high probability a random subgraph on $f(\eps)$-vertices will be $\tilde{\Omega}(\eps)$-far from bipartite. We know that $f(\eps)>1/K\eps$ for some absolute constant $K>0$ for instance from \cite{BT}. Conjecture 1 postulates that this bound is nearly tight, namely $f(\eps)=\tilde{O}(1/\eps)$.

We believe that Conjecture 1 is interesting in its own right. A similar phenomenon has been investigated in the context of probabilistically checkable proofs (PCPs) \cite{BSGH, DR}. Roughly, a PCP is {\em robust} if a proof which is $\eps$-far from a correct one must encode an assignment which is $\Omega(\eps)$-far from a satisfying one. Conjecture 1 says, in some sense, that bipartiteness is robust under takingrandom subgraphs on $\tilde{O}(1/\eps)$ vertices.

Even if the conjecture turns out to be false, we believe that coming up with a counterexample to it would be enlightening in understanding the complexity of testing bipartiteness.

\paragraph{Comparison with the work of Gonen and Ron} Gonen and Ron give two algorithms for testing bipartiteness that apply to two different classes of graphs. Their first algorithm runs in time $\tilde{O}(\eps^{3/2})$, but only applies to graphs where all degrees are at most $O(\eps n)$. Their second algorithm runs in time $\tilde{O}(1/k\eps^2)$, but only applies to graphs where all degrees are at least $k\eps n$.

As we explain below, the two algorithms of Gonen and Ron are fundamentally different from one another and there is no a priori reason why we should expect a single algorithm for testing bipartiteness in time $O(1/\eps^{2-c})$ that works for all graphs, In this work we provide such an algorithm, and we show that if Conjecture 1 is true, our algorithm is correct on all inputs.

\subsection{The algorithms of Gonen and Ron}
Our proof is inspired by the algorithms proposed by Gonen and Ron~\cite{GR} for testing bipartiteness in graphs with some restriction on the degrees of the vertices. Gonen and Ron propose two such algorithms: One for ``low-degree'' graphs and another one for ``high-degree'' graphs.

Before we explain the two algorithms, let us apply the conjectured generalization of the Alon-Krivelevich result. This generalization allows us to replace a graph on $n$ vertices by one on $\tilde{O}(1/\eps)$ vertices, while roughly preserving the distance from bipartiteness. 

Applying this conjecture comes at no loss in generality (if we disregard polylogarithmic factors in $1/\eps$): If testing bipartiteness {\em requires} $q(\eps)$ queries on graphs of size $s=\tilde{O}(1/\eps)$, then it requires the same number of queries on graphs on $n$ veritces whenever $n$ is a multiple of $s$. The counterexample of size $s$ can be scaled up to size $n$ by performing an $\eps n$-blowup of it.\footnote{A {\em $k$-blowup} of a graph $G$ is obtained by replacing every vertex of $G$ by a group of $k$ vertices and every edge of $G$ by a complete $k$ by $k$ bipartite graph between the corresponding groups of vertices.} It is not difficult to see that this blowup operation preserves the distance to bipartiteness.\footnote{One shows that the optimal partitions always assign the same color to all vertices within a group. For general graph properties, however, this is not always the case \cite{GKNR}.}

After applying the conjecture, let us use $n$ to denote the number of vertices in the induced subgraph. Then (with high probability) this graph $G$ will be $\tilde{\Omega}(1/n)$-far from bipartite. For convenience, we will make a change in terminology: We say that a graph is {\em $m$-removed} from bipartite if one cannot make it bipartite even after removing $m$ arbitrary edges. Then the graph $G$ will be $\tilde{\Omega}(n)$-removed from bipartite. We have now reduced the task to designing a tester for bipartiteness in $G$ that runs in time $O(n^{2 - c})$ for some constant $c$. This algorithm only needs to outperform the trivial algorithm by a small factor in the running time. However, when $G$ is far from bipartite, it is only guaranteed to have $\tilde{\Omega}(n)$ violating edges -- about one per vertex.

To explain the low-degree and high-degree algorithms, it is useful to keep in mind these two example graphs on $n$ vertices that are $\Omega(n)$-far from bipartite:
\begin{itemize}
\item $G_1$: A constant-degree expander -- for example, a collection of $d$ random matchings for some constant $d$.
\item $G_2$: A $\sqrt{n}$-blowup of an odd length cycle of length about $\sqrt{n}$.
\end{itemize}

\paragraph{The low degree algorithm}
When the maximum degree $d$ of $G$ is much smaller than $n$, we can test bipartiteness by emulating the Goldreich-Ron algorithm for bounded degree graphs. In the {\em adjacency lists} model with distance parameter $\delta$, the Goldreich-Ron algorithm runs in time $\tilde{O}(\sqrt{n})\cdot 1/\delta^K$, where $K$ is some constant. If a graph $G$ of maximum degree $d$ is $\tilde{\Omega}(n)$-removed from bipartite in the adjacency matrix model, then $G$ is $1/d\cdot \poly\log(n)$-far from bipartite in the adjacency list model, and we have the following consequence.

\begin{theorem}[The Goldreich-Ron algorithm]
\label{thm:gr}
Let $G$ be a graph on $n$ vertices maximum degree $d$. Suppose we are given an adjacency list representation of $G$. Then there is an algorithm that runs in time $\tilde{O}(\sqrt{n}) \cdot d^K$ such that (1) if $G$ is bipartite, the algorithm says ``bipartite'' and (2) if $G$ is $\tilde{\Omega}(n)$-removed from bipartite, the algorithm outputs an odd-length cycle of $G$ with probability $9/10$.
\end{theorem}

For intuition, we briefly outline the Goldreich-Ron algorithm: This algorithm chooses $O(1/\delta)$ random starting vertices $v$, then performs $\poly(\log(n)/\delta)\cdot \sqrt{n}$ random walks out of $v$ of length $\poly(\log(n)/\delta)$ and looks for pairs of random walks -- one of even length, the other one of odd length -- that start and end at the same vertex, thus revealing an odd-length cycle. For a graph like $G_1$, this algorithm works well because the random walks mix within $O(\log n)$ steps and we expect two of them to collide after about $\sqrt{n}$ attempts. We expect the length of the walks to be about evenly divided between even and odd, so in the first collision an odd-length cycle is likely to be revealed.\footnote{This crude explanation misses many important ideas of the Goldreich-Ron algorithm, and we refer the reader to their paper~\cite{GR} for a complete presentation.} 

In the theorem, $K$ is some universal constant which Goldreich and Ron do not attempt to optimize. Since every query in the adjacency list model can be emulated by $n$ queries in the adjacency matrix model, this algorithm runs in time $\tilde{O}(n^{3/2})/d^K$ for graphs given by adjacency matrices. Notice that the Goldreich-Ron algorithm outperforms the trivial algorithm not only when $d$ is constant, but also when $d$ is a sufficiently small power of $n$. However, when the degree becomes very large -- for example as in the graph $G_2$ -- then it is not difficult to see that the running time of the algorithm can easily become $\tilde{\Omega}(n^2)$ in the adjacency matrix representation.

\paragraph{The high degree algorithm}
Now assume the degrees of all the vertices are at least $d$. This algorithm chooses a random set $S$ of vertices of size about $O((n/d) \log n)$, queries all their neighbors for edges, makes another $O((n^2/d) \log n)$ random pair queries, and looks for an odd-length cycle. The running time of this algorithm is $\tilde{O}(n^2/d)$. Instead of explaining why this algorithm works in general, let's look at the representative graph $G_2$. Here, even in the first phase of the algorithm, we are likely to have sampled at least one vertex in every location of the cycle, so after querying all their neighbors we see at least one copy of the odd-length cycle.

On the other hand, the high-degree algorithm is unlikely to work well on graphs like $G_1$ with constant-degree vertices: After the first phase, the graph has seen many isolated vertices (and their neighbors) in the graph, and it is quite unlikely that a cycle will be revealed in the second phase.

\subsection{Our proof}
It follows from our discussion that there is some critical degree $d = n^a$, where $a$ is a constant, with the following property: If all vertices have degree less than $d$, then the low-degree algorithm outperforms the trivial one; and if all vertices have degree greater than $d$, then the high-degree algorithm outperforms the trivial one. So it is natural to try and combine both algorithms and get a single one that works for all graphs.

Now consider a graph $G$ that is $m = \tilde{\Omega}(n)$-removed from bipartite and has both low-degree and high-degree vertices. Intuitively, if $G$ is $m$-removed from bipartite, it could be for three reasons: Either there are many violating edges in the high-degree part of $G$, or there are many violating edges in the low-degree part, or there are many violating edges between the two parts. (To illustrate the third case, consider a graph that has $n/2$ vertices of degree about $\sqrt{n}$ that are connected by a blowup of a cycle of {\em even} length about $\sqrt{n}$, and $n/2$ vertices of constant degree that are randomly connected to these high-degree vertices.)

The high-degree algorithm lets us take care of the first case: If we run the high-degree algorithm on $G$ and it detects no odd cycle, then we can be fairly confident that the high-degree part of $G$ is close to bipartite. But in fact we prove that the high-degree algorithm does more: In the case the algorithm fails to detect an odd cycle, its queries reveal an {\em approximate spanning forest} of the high-degree part of $G$: That is, after removing $m/2$ edges, we know exactly which vertices in the high-degree part must be of the same color and which must be of different colors. We prove this in Section~\ref{sec:highdeg}.

Once we know this approximate partition of the high-degree part, we can try to run the low-degree algorithm, but ``short-cut'' its random walks once they enter the high-degree part. For instance, suppose we have two low-degree vertices $v, v'$ that connected, respectively, to high-degree vertices $w, w'$. If the queries of the high-degree algorithm revealed that $w$ and $w'$ are of different colors, then we can redirect a random walk that attempts to go from $v$ to $w$ directly into $w'$. 

To be a bit more formal, the process we just described is an {\em emulation} of the graph $G$ by a new graph $G'$. The graph $G'$ will contain only those vertices of $G$ that have low degree, all the edges between them, and some extra edges that are meant to capture the connectivity information among the high-degree vertices of $G$. For the emulation to be correct, it must satisfy four requirements: (1) If $G$ is bipartite, so is $G'$; (2) if $G$ is far from bipartite, so is $G'$; (3) queries in $G'$ can be emulated efficiently using queries in $G$; and (4) even after adding the extra edges, all vertices in $G'$ remain of low degree.

This emulation of $G$ by $G'$ is explained in detail in Section~\ref{sec:eliminate}. Ignoring many relevant issues, the basic idea is to replace each ``component'' of the high-degree part of $G$ revealed in the first phase by a {\em random graph} of (about) constant degree on the low-degree vertices it connects to which is consistent with the partition. Since a random graph is likely to be an expander, it turns out that this transformation preserves distance from bipartiteness. On the other hand, since edges in a random graph can be generated ``on the fly'' by sampling, the emulation can be performed in a query-efficient manner.

For somewhat technical reasons, it will be easier for $G'$ to be an instance of a problem which is a generalization of bipartiteness. This problem, which we call an {\em XOR game} (borrowing from PCP terminology), is a binary constraint satisfaction problem with two kinds of edge-constraints: $\neq$-constraints and $=$-constraints which say that the colors of the endpoints should be different and equal, respectively.

Finally, in Section~\ref{sec:bipartite} we show that an XOR game can be transformed into an ordinary instance of bipartiteness without significantly affecting the degree and the query complexity, while roughly preserving the distance. 

To summarize, the transformations we have described allow us to take a graph which can have vertices of arbitrary degrees and either locate an odd-length cycle in the high-degree graph, or transform the graph into a low-degree graph which roughly preserves distance from bipartiteness. All these transformations are done in a query-efficient way. We prove Theorem~\ref{thm:main} in Section~\ref{sec:main}.

\subsection{Notation}
A graph $G$ is {\em $m$-removed from bipartite} if even after removing any subset of at most $m$ edges from $G$, the resulting graph is not bipartite. A graph is {\em $\eps$-far from bipartite} (in the dense graph model) if it is $\eps\binom{n}{2}$-removed from bipartite.

For a graph $G$ and a subset of vertices $S$, we write $G_S$ for the induced subgraph of $G$ on $S$.

Let $G$ be a graph, and $G'$ a subgraph of $G$ (possibly on a subset of the vertices). A {\em spanning forest} of $G'$ in $G$ is a forest $F$ of $G$ such that every two vertices in the same component of $G'$ are connected by a path in $F$. When $G' = G_S$, abusing terminology we will write ``spanning forest of $S$ in $G$.''

An {\em XOR game} is a graph where each edge is labeled by $=$ or $\neq$. The game is {\em satisfiable} if there is a 2-coloring of the vertices that satisfies all the edge labels, and it is {\em $m$-removed from satisfiable} if it cannot be made satisfiable even after removing any set of at most $m$ labels.

We use $E_G(S, T)$ for the set of edges between $S$ and $T$ in $G$. We say $G$ is an {\em $\alpha$-expander} if $\abs{E_G(S, \Sbar)} \geq \alpha\abs{S}$ for every set $S$ containing at most half the vertices.

In the {\em dense graph model} we are given access to the adjacency matrix of the graph. In the {\em sparse graph model} we are given adjacency lists for every vertex.

\section{Splitting the vertices by degree}
\label{sec:split}

We give a simple, sampling-based algorithm that splits the vertices of $G$ into those whole degree is roughly smaller than $d$ and those whose degree is roughly larger than $d$. This algorithm runs in time $\tilde{O}(n^2/d)$.

\begin{quote}
{\bf Algorithm Degree}: For every vertex $v$, query $(v, w)$ for $24n\log n/d$ random $w$s. If more than $24\log n$ edges are detected, put $v$ in $H$. Otherwise, put $v$ in $L$.
\end{quote}

\begin{claim}
With probability at least $9/10$, every vertex in $H$ has degree at least $d/2$, and every vertex in $L$ has degree at most $2d$.
\end{claim}

The claim follows easily from a large deviation bound. Let $d(v)$ be the degree of $v$. Each $w$ is a neighbor of $v$ independently at random with probability $d(v)/n$, so the expected number of neighbors among the queried $w$s is $\log n \cdot d(v)/d$. By a Chernoff bound we get that the number of neighbors is in the range $(24 \pm 12) \log n (d(v)/d)$ with probability $1 - 1/10n$. If this event holds for a particular vertex, then a vertex of degree more than $2d$ cannot end up in $L$, and a vertex of degree less than $d/2$ cannot end up in $H$. The conclusion follows by a union bound over all $v$.

\section{The algorithm for high degree vertices}
\label{sec:highdeg}

We now describe the algorithm for the high-degree part of the graph. The goal of this algorithm is to either detect a violation of bipartiteness in this part of the graph, or find an approximate spanning forest.

\begin{quote}
{\bf Algorithm High-degree}: Choose $s = 10n\log n/d$ random vertices $v$ of $G$ and query $(v, w)$ for every $w$.
Query an additional $t = 40n^3(\log n)^2/md$ random pairs $(u, u')$.
\end{quote}

\begin{claim}
\label{claim:highdeg}
Assume $G$ is $m$-removed from bipartite. Let $Q$ be the set of queries performed by Algorithm 1. Then with probability at least $9/10$, either $Q$ contains an odd cycle, or there is a set of at $m/2$ edges such that after removing these edges from $G$, we obtain a graph $G''$ such that $G''_H$ is bipartite and $Q$ contains a spanning forest of $G''_H$ in $G$.
\end{claim}

Let $S$ be the set of $n\log n/d$ random vertices $v$. With probability $19/20$, $S$ dominates $H$: The probability that any fixed vertex of $H$ is not dominated is at most $(1 - d/2n)^s \leq 1/20n$, so by a union bound $S$ dominates $H$ with probability $19/20$.

Now assume $S$ dominates $H$. We split the analysis in two cases. First, consider what happens if $G_H$ is $m/4$-removed from bipartite. Then we show that with probability $19/20$, $Q$ will contain an odd cycle. Consider any fixed bipartition of $S$. This partition induces a bipartition on $H$, which must have at least $m/4$ violating edges. The probability that one of the random pairs $(u, u')$ hits a violating edge is then $(1 - m/4n^2)^t < 2^{-s}/20$. By a union bound over all partitions of $S$, it follows that all of the induced partitions on $H$ have a violating edge in the sample with probability at least $19/20$. If this is the case, $Q$ must contain an odd cycle.

Now let's assume $G_H$ is $m/4$-close to bipartite. Delete at most $m/4$ violating edges to make it bipartite and call the resulting graph $G'$. We now focus attention on $G'$ and $G'_H$.

Consider the following partition $\mathcal{H} = \{H_1,\dots,H_h\}$ ($h \leq s$) of $H$ induced by $S$: For every $v_i \in S$, let $H_i$ be the set of those vertices in $H$ connected to $v_i$ in $G'$ but not to any of $v_1,\dots,v_{i-1}$. We show that with probability $19/20$, any cut $C$ that preserves $\mathcal{H}$ (i.e. $C$ does not split any set $H_i$) with at least $m/4\ceil{\log n}$ edges in $G'_H$  contains an edge sampled in the second stage. Consider a fixed such cut $C$. The probability that no edge in $C$ is picked in the second stage is at most $(1 - m/n^2\log n)^t < 2^{-s}/20$. We conclude by taking a union bound over all such $C$, as there are at most $2^s$ of them.

Now assume that every cut $C$ that preserves $\mathcal{H}$ has either an edge in $Q$ or at most $m/4\ceil{\log n}$ edges in $G'_H$. Look at any maximal forest inside $Q \cap E(G')$. The partition of $H$ into components induced by this forest (i.e. two vertices in $H$ are in the same component if they are connected in $Q$) is a coarsening of the partition $\mathcal{H}$. Let's call it $\mathcal{H'}$. So every cut $C$ that preserves $\mathcal{H'}$ has at most $m/4\log n$ edges in $G'_H$. Now take any collection of $\ceil{\log_2\abs{\mathcal{H'}}}$ such cuts that shatter $\mathcal{H'}$ (that is, every pair of sets in $\mathcal{H'}$ is partitioned by at least one of the cuts). Removing all the edges in all these cuts will disconnect every pair of sets of $\mathcal{H'}$ in $G'_H$. The total number of removed edges is at most $(m/4\ceil{\log n}) \cdot \ceil{\log_2\abs{\mathcal{H'}}} \leq m/4$. After removing all these edges we obtain a new graph $G''$ such that $Q$ becomes a spanning forest of $G''_H$ in $G$. Moreover, $G''$ is obtained from $G$ by removing at most $m/4 + m/4 = m/2$ edges.

\section{Eliminating the high degree vertices}
\label{sec:eliminate}

We come to the step of replacing the high-degree vertices of $G$ with ``equivalent'' edges in the low-degree part of $G$. The idea is to replace each (sufficiently ``well-connected'') component in the high-degree part of the graph with an expander graph that has the same connectivity properties (namely, edges of the expander go only among vertices that are on opposite sides of the partition induced by the trees in the spanning forest). Let $G'$ denote the new graph. The vertices of $G'$ will be only those vertices of $G$ which are in $L$. However, in addition to the edges in $L$, $G'$ will have some additional edges which are meant to describe the connectivity within the high-degree part of $G$.

We need to do this transformation of $G$ into $G'$ in a manner that preserves the soundness of the instance (a graph that is far from bipartite stays far from bipartite), but does so in a query-efficient manner. Specifically, we want to make sure that every query in the $G'$ can be simulated by a small number of queries in the original graph. To do so, we make a distinction between two types of high-degree components in $G$: Those that have high connectivity to the low-degree part and those that don't. 

Our construction of $G'$ is probabilistic: $G'$ will be a random graph which is guaranteed to meet the specified construction and answer queries in $G$ efficiently with high probability.

If a component $C$ of high-degree vertices has fewer than $\tilde{O}(\abs{C})$ edges going into the low-degree part, then we can safely ignore it because disconnecting all such components in a graph which is far from bipartite keeps the graph far from bipartite.

If a component has more than this number of edges going to the low-degree part, then we want to {\em emulate} this component in $G'$ by an expander graph among the vertices of the low-degree part consistent with the high-degree component. However, we do not know the structure of the $G$ in advance (since we want an algorithm that runs in sublinear time), so we need to implement this expander graph ``on the fly''. 

To be more specific, for each high-degree component $C$ and each edge going out of this high-degree component to a low-degree vertex $v$, we want to introduce an expander edge going out of $v$ in $G'$. We will construct this expander {\em at random}: To emulate an edge going from a low-degree vertex $v$ to a high-degree vertex $w$ in some component $C$ in $G$, we want to choose a random edge between some other $w' \in C$ of the opposite color and some $v' \in L$ and connect $v$ to $v'$. The key point is that since the cut between $C$ and the low-degree vertices is {\em dense}, such an edge can be found by random sampling in a relatively query-efficient way.

To summarize, we replace each high-degree component $C$ of $G$ which has more than about $\abs{C}$ edges going into the low-degree part by a random graph (of small degree) between its low-degree neighbors. Since a random graph is likely to be an expander, this step is likely to preserve distance from bipartiteness. On the other hand, since by assumption the cut between $C$ and $L$ is fairly dense, the step can be implemented in a query-efficient way.

The actual proof includes some additional technicalities. For example, some extra work needs to be done in order to determine if a component $C$ has a sparse cut or a dense cut to $L$. We do so by a preliminary sampling step. 

Also, there could be an imbalance between the number of vertices of the two colors within a component $C$, in which case the sampling of an exit edge may take too long. To resolve this, we will allow to have edges within $G'$ between vertices of the same color, but we will mark those edges as $=$-edges (to indicate that the colors of the endpoints are equal and not different). This gives a generalized version of bipartiteness: A constraint satisfaction problem with two types of constraints which indicate the colors of the endpoints are equal or not equal. We will allow for these generalized constraints and replace them with ordinary (inequality) constraints later on.

\paragraph{Specification of $G'$.} Assume $G_H$ is bipartite. Let $\calC$ be the collection of connected components of $G_H$. We say $G'$ is a {\em sparse version} of $G$ if the following conditions hold:
\begin{enumerate}
\item The set of vertices of $G'$ is $L$.
\item For every edge $(v, w)$ in $G$ with $v, w$ in $L$, $(v, w)$ is an $\neq$-edge of $G'$.
\item Assume $C \in \calC$ satisfies $\abs{E_G(C, L)} \geq \abs{C}m/3n$. Let $Z$ be an arbitrary $\Delta$-regular $1/2$-expander on $\abs{E_G(C, L)}$ vertices. Each vertex of $Z$ represents an edge $(v, w)$, where $v \in L$ and $w \in C$. Then we have the following edges in $G'$: For every pair $(v, w)$, $(v', w')$, where $v, v' \in L$, $w, w' \in C$, if $(v, w)$, $(v', w')$ is an edge in $Z$, then we put an edge between $v$ and $v'$ in $G'$. This is an $=$-edge if $w$ and $w'$ are at even distance in $G_H$, and a $\neq$-edge otherwise.
\item Assume $C \in \calC$ satisfies $\abs{E_G(C, L)} \leq \abs{C}m/6n$. For every edge $(v, w)$ such that $v \in L$, $w \in C$, there is an $=$-loop $(v, v)$ in $G'$.
\item If $\abs{C}m/6n < \abs{E_G(C, L)} < \abs{C}m/3n$, then both are allowed.
\end{enumerate}

Notice that if $G_L$ has max-degree $2d$, then $G'$ has max-degree $4\delta d$.

\begin{claim}
\label{claim:spec}
Assume $G_H$ is bipartite and $G'$ is a sparse version of $G$. Then (1) If $G$ is bipartite, then $G'$ is satisfiable and (2) if $G$ is $m$-removed from bipartite then $G'$ is $m/3$-removed from satisfiable.
\end{claim}

Part (1) follows easily as any consistent bipartition of $G$ is by construction a consistent assignment of $G'$.

For (2) we prove the contrapositive. Assume we have a partition $(S, \Sbar)$ of $L$ that violates at most $m/3$ constraints of $G'$. Write the number of violated constraints as $m_L + \sum_{C \in \calC} m_C$, where $L$ is the number of violated constraints within $L$ and $m_C$ is the number of violated constraints owing to component $C$.

We show $(S, \Sbar)$ can be extended to a partition of $L \cup H$ that violates at most $m$ edges of $G$. We define the coloring on $H$ as follows: Since $G_H$ is bipartite, for every component $C$ of $G_H$ there are two possible partitions. Among the two, choose the one that violates the smaller number of edges between $C$ and $L$.

We now upper bound the number of violated edges in $G$ with respect to this partition. Within $L$, the number of violated edges is exactly $m_L$. Within $H$, by construction there are no violated edges. So it remains to bound the number of violated edges between $L$ and $C$ for every $C \in \calC$. For those $C$ in which the construction in step (4) was applied, we have introduced at most $\abs{E_G(C, L)} \leq \abs{C}m/3n$ violating edges. Even summing over all such components $C$ this introduces at most $m/3$ new violating edges.

It remains to bound the contribution of those $C$ for which the expander construction (3) was applied.  
Let $(T_C, \Tbar_C)$ be the partition of the vertices of $Z$ induced by the coloring on $H$: The color of $e \in Z$ is determined by the color of its endpoint inside $C$. Let $(S_C, \Sbar_C)$ be the partition of the vertices of $Z$ induced by the coloring on $L$: Here the color of $e \in V(Z)$ is determined by the color of its endpoint inside $L$. Notice that a constraint $(e, e')$ of $Z$ is then violated if and only if the $(T_C, \Tbar_C)$ and $(S_C, \Sbar_C)$ of $(e, e')$ are inconsistent, namely
\[ \abs{(S_C(e) \oplus S_C(e')) \oplus (T_C(e) \oplus T_C(e'))} = 1. \]
(Here $A(e) = 1$ if $e \in A$, $0$ if $e \not\in A$.) We can rewrite this equation as
\[ \abs{(S_C(e) \oplus T_C(e)) \oplus (S_C(e') \oplus T_C(e'))} = 1. \]
Therefore, a pair $(e, e')$ violates a constraint of $Z$ if and only if it cuts across the partition $(S_C\oplus T_C, \overline{S_C \oplus T_C})$. Assume without loss of generality that $S_C \oplus T_C$ is the smaller of the two sets. By the expansion of $Z$, the number of violated constraints $m_C$ of $G'$ is at least as large as $\abs{S_C \oplus T_C}/2$.But $\abs{S_C \oplus T_C}$ is exactly the number of edges between $L$ and $H$ whose endpoints are colored inconsistently.

We summarize the total number of violated edges in $G$: We have $m_L$ such edges inside $L$, at most $m/3$ edges coming from type (4) components and at most $2m_C$ edges for every type (3) component $C$. So the total number of violating edges is at most
\[ m_L + m/3 + 2\sum\nolimits_\text{$C$ of type (3)} 2m_C \leq m. \]

\paragraph{Implementation of a $G'$-oracle given oracle access to $G$.} This is a randomized procedure Reduce which is a query-efficient implementation of $G'$ and is a sparse version of $G$ with high probability. This procedure has oracle access to $G$ and direct access to a spanning forest $F$ of $G_H$. For each connected component $C$ of $G_H$, Reduce will have a variable $s_C$ that has value ``sparse'' (indicating that the $\abs{E_G(C, L)}$ is small) or ``dense'' (indicating it is large). Initially, all these variables are uninitialized.

Without loss of generality we will assume $\abs{E_G(C, L)}$ is an even number. This can be enforced, say, by making two copies of every edge in $G$.

The idea of Reduce is to build (in an online manner) an expander $Z$ on the dense components, while replacing the sparse components with loops. 
The expander $Z$ will consist of $\Delta$ random perfect matchings. These matchings are constructed in an online manner. For every component $C$ of $G_H$ that has been declared dense, the implementation of $G'$ keeps track of the partial perfect matchings corresponding to this component. (Initially, when $s_C$ is uninitialized the collection of matchings is empty.

It will be more convenient to specify $G'$ in the sparse graph model (via adjacency lists). Remember $G$ is provided to us in the dense graph model. A query of the type $(v, i)$ should be answered by $(w, =)$ (if $v$ and $w$ are connected by an $=$-edge) $(w, \neq)$ (if $v$ and $w$ are connected by a $\neq$-edge) or $\bot$ (if $v$ has fewer than $i$ neighbors).

Here is how a query $(v, i)$ for the $i$th neighbor of $v \in L$ is answered in $G'$:

\begin{enumerate}
\item Query $(v, w)$ in $G$ for all $w$ to find all neighbors of $v$ in $G$. There are at most $4d$ of them (remember we doubled the number of edges). Make two lists: One that contains those neighbors in $L$, and the one that contains those neighbors in $H$. Let $k$ be the number of neighbors in $H$.
\item If $i > k\Delta$, return $(w, \neq)$, where $w$ is the $(i - k\Delta)$-th neighbor of $v$ on the $L$-list (or $\bot$ if it doesn't exist).
\item If $i \leq k\Delta$, let $w$ be the $\floor{i/\Delta}$-th neighbor of $v$ in $H$ and $j = i \mod \Delta$.
\item If $w$ belongs to a component $C$ such that $s_C$ is uninitialized, repeat the following experiment for $100n\abs{L}\log n/m$ steps: Choose random vertices $u \in C$ and $z \in L$ and query if $(u, z)$ is an edge in $G$. If the number of detected edges is less than $100\log n/4$ set $s_C = {\rm ``sparse"}$ (and go to step 5), otherwise set $s_C = {\rm ``dense"}$ (and go to step 6).
\item If $w$ belongs to a component $C$ of $H$ such that $s_C = {\rm ``sparse"}$, return $(v, =)$ (Self-loops.)
\item If $w$ belongs to a component  $C$ of $H$ such that $s_C = {\rm ``dense"}$: If this edge was already visited, return the previous answer. Otherwise, repeat the following procedure until an output is produced: Choose random vertices $u \in C$ and $z \in L$. If $(u, z)$ is an edge in $G$: If the $j$th matching edge corresponding to $(u, z)$ has not been visited yet, mark that the $j$th matching connects $(v, w)$ and $(u, z)$ (i.e., mark that $(v, i)$ goes to $u$ and mark that the corresponding edge out of $u$ with index $\text{(order of $w$)}\Delta + j$ goes to $v$). Output $(u, =)$ if $w$ and $u$ are connected by a path of even length in $F$, and $(u, \neq)$ if they are connected by a path of odd length. 
\end{enumerate}

The {\em amortized query complexity} of the $q$th query into $G'$ is the total number of queries $G'$ makes into $G$ while answering its first $q$ queries divided by $q$.

\begin{claim}
\label{claim:compl}
(Amortized query complexity analysis) Assume $m = \tilde{\Omega}(n)$. Fix any sequence of queries to $G'$. With probability $9/10$, the amortized query complexity of the last query into $G'$ is $\tilde{O}(n)$.
\end{claim}

\begin{claim}
\label{claim:corr}
(Correctness of implementation) Assume $G_H$ is bipartite and $F$ is a spanning forest of $H$. Further assume $\Delta \geq 3\log n$. With probability $9/10$, $G'$ is a sparse version of $G$.
\end{claim}

We explain the somewhat unusual assumption on $\Delta$ below. To prove both of these claims, we start with the simple claim that with probability $19/20$ procedure Reduce correctly identifies the sparse and dense cuts between $L$ and $H$:

\begin{claim}
\label{claim:auximpl}
With probability at least $19/20$, the following is true for all components $C$: If $\abs{E_G(C, L)} \geq \abs{C}m/3n$ then $s_C = {\rm ``dense"}$ and if $\abs{E_G(C, L)} \leq \abs{C}m/6n$ then $s_C = {\rm ``sparse"}$.
\end{claim}
\begin{proof}
Let $s = 100n\abs{L}\log n/m$. Suppose $\abs{E_G(C, L)} \geq \abs{C}m/3n$. Then each experiment in round (4) independently hits an edge of $G$ with probability $p = \abs{E_G(C, L)}/\abs{C}\abs{L}$. The expected number of hits is at $ps \geq 100\log n/3$. By a Chernoff bound, the probability of getting fewer than $3/4$ the expected number of hits is less than $1/20n$. But if we get at least $3/4$ the desired number of hits, $s_C$ is correctly set to ``dense''. Similarly if $\abs{E_G(C, L)} \leq \abs{C}m/3n$, we get an expected number of no more than $ps \leq 10\log n/6$ hits, and the probability of getting more than $3/2$ the expected number of hits is less than $1/20n$. But if we get at most $3/2$ the expected number of hits then $s_C$ is correctly set to ``sparse''. Taking a union bound over all $C$ (there are at most $n$ of them) we get the desired result.
\end{proof}

To prove Claim~\ref{claim:compl}, we need to argue that every query in $G'$ can be emulated by $\tilde{O}(n)$ queries in $G$. We reason like this. On query $(v, i)$, we take time $O(n)$ to list all the neighbors of $v$ in $G$. If the $i$th neighbor happens to be in $L$ (step 2 in algorithm Reduce), 
we have answered the query in $O(n)$ time. Otherwise, if $s_C$ is uninitialized, we take another $\tilde{O}(n)$ time to initialize it. So if $C$ is sparse, we answer the query in $\tilde{O}(n)$ time. We now come to the more interesting case where $C$ is dense.

Consider what happens when we query an edge in the $j$th matching in some dense component $C$. If this edge was queried previously, we are done in time $\tilde{O}(n)$. But otherwise, we have to find a random edge between $C$ and $L$ previously unvisited in the $j$th matching. The more edges we see in the $j$th matching, the harder the new ones become to find. This is essentially the coupon collecting problem, so we expect to visit the edges in the $j$th matching in amortized time $O(\log n)$. Since it takes extra time $\tilde{O}(n)$ to just locate a single random edge between $C$ and $L$, we get the desired amortized query complexity.

\begin{proof}[Proof of Claim~\ref{claim:compl}]
By Claim~\ref{claim:auximpl} with probability $19/20$ if $s_C = {\rm ``dense"}$, it must be that $\abs{E(C, L)} \geq \abs{C}m/6n$. We assume that this is the case. Let $T_i$ denote the time taken to answer the $i$th query, $1 \leq i \leq q$. The quantities $T_i$ are random variables and we will bound the expectation of their sum.

Without loss of generality we will assume that no queries are repeated. Also, if the $j$th neighbor of $v$ was already determined to be the $j'$th neighbor of $u$, we will assume the query $(u, j')$ is not repeated. In all these cases the answer can clearly be determined in $O(n)$ time (and even much less).

We now group the queries into a set $R$ and sets $S_{C,j}$, where $R$ are queries within $L$
or between $L$ and sparse components of $C$ and $S_{C, j}$ are queries that involve the $j$th matching inside dense component $C$. By linearity of expectation,
\[ \E[T_1 + \dots + T_q] = \sum_{k \in R} \E[T_k] + \sum_{C,j} \sum_{k \in S_{C,j}} \E[T_k]. \]
As we argued above, each entry in the first some is bounded by $Kn^2\log n/m$ for some constant $K$. We now consider the entries in the second sum. Fix $C$ and $j$ and let $k_1 \leq \dots \leq k_t$ be the order of the queries in $S_{C, j}$. Let $e = \abs{E(C, L)}$. Since the queries do not repeat, $t \leq e/2$. In order to answer query $k_{i+1}$ (once $k_1,\dots,k_i$ have been answered), we need to locate an edge between $C$ and $L$ which is not already matched in the $j$th matching. The number of already matched edges in the $j$th matching is $2i$. This leaves $e - 2i$ good edges between $C$ and $L$ that can be used for the $j$th matching out of $v$. The probability that such an edge is hit by a random pair $(u \in C, z \in L)$ is $(e-2i)/\abs{C}\abs{L}$. Therefore, the expected number of attempts before such a pair is hit is $\abs{C}\abs{L}/(e - 2i)$.

It follows that the expected time $\E[T_{k_i}]$ to perform the $i$th query is at most $K\abs{C}\abs{L}/(e - 2i)$, where $K$ is some fixed constant. Adding all of these together we get that
\[ \sum_{k \in S_{C,j}} \E[T_k] \leq K \sum_{i=1}^t \frac{\abs{C}\abs{L}}{e - 2i}
   \leq K  \sum_{i=1}^{t} \frac{\abs{C}\abs{L}}{e - i} =  K\abs{C}\abs{L}  \sum_{i=1}^{t} \frac{1}{e-t + i}
   \leq 2K\abs{C}\abs{L} \log e, \]
using the standard estimate of the harmonic series. Since $C$ is dense, we have that
\[ \sum_{k \in S_{C,j}} \E[T_k] \leq 2K \abs{L} \cdot \frac{6ne}{m} \cdot \log e = O((n^2 \log n /m) \cdot t) \]
because $t \leq e/2$. Using the assumption that $m = \tilde{\Omega}(n)$ and Markov's inequality, we obtain the desired amortized query complexity.
\end{proof}

To finish the proof of Claim~\ref{claim:corr} we additionally need to show that for each component $C$ that was identified as dense, the simulated queries in step (6) are consistent with an expander graph on the set $Z$. To prove this, notice that algorithm Reduce implements $Z$ by a collection of $\Delta$ random matchings. It is a well-known fact that such a collection of random matchings is an expander graph with high probability. In order to be able to take a union bound over all components $C$, we need to make the probability of each cut large enough in $n$, and this is why we need to assume that $\Delta \geq 3\log n$ (as supposed to constant as is standard in such calculations).

\begin{proof}[Proof sketch of Claim~\ref{claim:corr}]
By Claim~\ref{claim:auximpl} the classification of components into sparse and dense is correct with high probability, so it remains to argue that with high probability for all dense components $C$, all graphs constructed in step (6) are expanding. Let's fix such a dense component $C$ and let $Z$ be the corresponding graph. We will prove that the probability $Z$ is not $1/2$ expanding is at most $1/20n$. The conclusion then follows by a union bound over all $C$.

Let $z = \abs{Z}$ and fix a partition $(S, \Sbar)$ of $Z$ with $s = \abs{S}$ and $s \leq z/2$. We first consider the case $s \leq z/4$. Start choosing the edges in the first $\Delta/2$ matchings  going out of $S$ one by one. At each point, the probability that the other endpoint of the edge lands inside $S$ is at most $2s/z$. Taking a union bound over all $S$ of size $s$, we get that the probability that fewer than $s/2$ edges land outside $S$ is at most 
\[ \binom{z}{s} \cdot \binom{\Delta s/2}{s/2} \cdot \Bigl(\frac{2s}{z}\Bigr)^{\Delta s/2 - s/2} 
   \leq \Bigl(2e^3\Delta(2s/z)^{\Delta-3}\Bigr)^{s/2} \leq (16e^3\Delta/2^\Delta)^{s/2}, \]
using the inequality $\binom{n}{k} \leq (en/k)^k$ and the assumption $s \leq z/4$. Now we use the fact that $\Delta \geq 3\log n$ and conclude that the probability that {\em some} cut of size at most $z/4$ is violated is no more than
\[ (16e^3\Delta/2^\Delta)^{1/2} + \sum_{s=2}^{z/4} (16e^3\Delta/2^\Delta)^{s/2} \leq 1/40n \]
for sufficiently large $n$. Now suppose $z/4 \leq s < z/2$. In this case, the expected number of edges in the cut $(S, \Sbar)$ is $\Delta s(z - s) \geq \Delta s/2$. The events that an edge is in the cut are not independent, but they form a martingale, so we can apply Azuma's inequality (see for instance~\cite{BOT} for a similar application of Azuma's inequality) to conclude that fewer than half of these edges make it into the cut with probability at most $e^{-\Delta s/8} \leq e^{-\Delta z/32} \leq 2^{-z}/40n$, using the assumption $\Delta \geq 3\log n$. Taking a union bound over all such cuts we get an extra contribution of $1/40n$, so the probability $Z$ is not expanding is at most $1/20n$. This concludes the proof.
\end{proof}

\section{From an XOR game to a bipartite graph}
\label{sec:bipartite}

We give a bounded degree reduction from XOR games to bipartiteness. Assume $G$ is an XOR game with $n$ vertices of maximum degree $D$. Consider the following graph $G'$:

\begin{enumerate}
\item $G'$ has $2n$ vertices. Each vertex $v$ of $G$ is represented by a pair $v, v'$.
\item For every $\neq$-edge $(v, w) $in $G$, we put the edges $(v, w)$ and $(v', w')$ in $G'$.
\item For every $=$-edge $(v, w)$ in $G$, we put the edges $(v, w')$ and $(v', w)$ in $G'$.
\item For every vertex $v$ in $G$, we put in $D$ parallel edges $(v, v')$ in $G'$.
\end{enumerate}

Notice that $G'$ is a graph on $2n$ vertices with maximum degree $2D$. Moreover, given oracle access to $G$, a query in $G'$ can be implemented in constant time (in the bounded degree model).

\begin{claim}
\label{claim:xor}
Assume $G$ is $m$-removed from satisfiable. Then $G'$ is $2m$-removed from bipartite.
\end{claim}
\begin{proof}
We show the contrapositive. Assume $G'$ is $2m$-close to bipartite. Consider the optimum partition of $G'$. We think of it as a disjoint union of two partitions $(S \cup \overline{T}, \overline{S} \cup T)$, where $(S, \overline{S})$ is a partition of the $v$-vertices, and $(\overline{T}, T)$ is a partition of the $v'$ vertices. We will show that the better one of the partition $(S, \overline{S})$ and $(T, \overline{T})$ violates no more than $m$ edges in $G$. 

Let $e_=(U)$ be the total number of $=$-edges of $G$ with both endpoints in $U$. Let $e_=(U, V)$ be the total number of pairs $(u, v)$ where $u \in U$, $v \in V$, and $(u, v)$ is n $=$-edge of $G$. (We allow for $U$ and $V$ to intersect.) We define $e_{\neq}(U)$ and $e_{\neq}(U, V)$ analogously.

The number of constraints in $G$ violated by both $(S, \overline{S})$ and $(T, \overline{T})$ in is then
\[ \bigl(e_{\neq}(S) + e_{\neq}(\overline{S}) + e_=(S, \overline{S})\bigr)  + \bigl(e_{\neq}(T) + e_{\neq}(\overline{T}) + e_=(T, \overline{T})\bigr). \]
On the other hand, the number of edges of $G'$ violated by the partition $(S \cup \overline{T}, \overline{S} \cup T)$ is
\[ \bigl(e_{\neq}(S) + e_{\neq}(\overline{S}) + e_{\neq}(T) + e_{\neq}(\overline{T})\bigr)
  + \bigl(e_=(S, \overline{T}) + e_=(\overline{S}, T)\bigr) + D\abs{S \oplus T}, \]
where $S \oplus T$ is the symmetric difference of $S$ and $T$. We now argue that
\[ e_=(S, \overline{S}) +  e_=(T, \overline{T}) \leq e_=(S, \overline{T}) +  e_=(\overline{S}, T) + D\abs{S \oplus T}. \]
We have
\[ e_=(S, \overline{S}) = e_=(S, \overline{S} \cap \overline{T}) + e_=(S, \overline{S} \cap T)
   \leq e_=(S, \overline{T}) + e_=(S, \overline{S} \cap T) \]
and by symmetry
\[ e_=(T, \overline{T}) \leq e_=(T, \overline{S}) + e_=(T, \overline{T} \cap S). \]
Since $\overline{S} \cap T$ and $\overline{T} \cap S$ are a disjoint partition of $S \oplus T$, we get
\[ e_=(S, \overline{S} \cap T) + e_=(T, \overline{T} \cap S) \leq D\abs{S \oplus T} \]
which proves the desired inequality. It follows that together the partitions $(S, \overline{S})$ and $(T, \overline{T})$ violate at most $2m$ edges of $G$, so one of them will violate at most $m$ edges of $G$.
\end{proof}

\section{Proof of the main theorem}
\label{sec:main}

We now describe how the above steps combine into an algorithm for testing bipartiteness. Let $K$ be the constant from the Goldreich-Ron algorithm. On input a graph:
\begin{enumerate}
\item Choose a random subset of $\tilde{O}(1/\eps^2)$ vertices of and denote the induced subgraph on those vertices by $G$. Let $n$ be the number of vertices of $G$ and set $m = n/(log n)^c$ where $c$ is a sufficiently large constant. Set $d = n^{3/2(K+1)}$ and $\Delta = 3\log n$.

\item Run algorithm Degree on $G$.

\item Run algorithm High-degree on $G$. If the algorithm produces an odd-length cycle, reject. Otherwise, let $F$ be any maximal subforest of the queries. 

\item Let $G_1$ be the (oracle) XOR game obtained from $G$ and the forest $F$ using the implementation in Section~\ref{sec:eliminate}.

\item Let $G_2$ be the (oracle) graph obtained from $G_1$ using the XOR-to-bipartiteness reduction in Section~\ref{sec:bipartite}.

\item Run the Goldreich-Ron bipartiteness algorithm on $G_2$ and return its answer.
\end{enumerate}

\paragraph{Completeness analysis} The completeness of this algorithm follows by construction: Any bipartition of the original graph carries over to the graphs $G_1$, $G_2$ and $G_3$, so if the original graph is bipartite the algorithm will always say so.

\paragraph{Soundness analysis} We now argue soundness. We assume the original graph is $\eps$-far from bipartite. By Conjecture~\ref{conj:ak}, with probability $9/10$, $G$ is $m$-removed from bipartite. If this is the case, then by Claim~\ref{claim:highdeg}, with probability $9/10$, unless the algorithm has found an odd cycle in step 2, there exists a subgraph $G'$ of $G$ obtained by removing at most $m/2$ edges such that $F$ is a spanning forest of $H$ in $G'$ and $G'$ restricted to $H$ is bipartite. 

The set $F$ could contain edges that are not in $G'$; add these to $G'$ to obtain a graph $G''$. Then $F$ remains a spanning forest of $G''$ and $G''$ restricted to $H$ is still bipartite (any bicoloring induced by $F$ is good for both $G_H'$ and $G_H''$). Moreover, $G_H''$ is still $m/2$-removed from bipartite (since we only added edges to it).

By Claim~\ref{claim:corr}, with probability $9/10$, $G_1$ is then a sparse version of $G''$. By Claim~\ref{claim:spec}, if this is the case, then $G_1$ is $m/6$-removed from bipartite. Finally, by Claim~\ref{claim:xor}, $G_2$ is $m/3$-removed from bipartite. Then the Goldreich-Ron algorithm will detect an odd cycle with probability $9/10$. Taking a union bound over all failure events we get that an odd-length cycle is found with probability at least $7/10$.

\paragraph{Running time analysis} The running time analysis is straightforward. Steps 1, 2, and 3 all run in time $\tilde{O}(n^2/d)$. For step 4, by Claim~\ref{claim:compl} with probability $9/10$ the running time of the remaining part is $\tilde{O}(n) \cdot q$, where $q$ is the number of queries performed by the Goldreich-Ron algorithm. Notice that the maximum degree of vertices in $G_2$ is $2\Delta d \leq 6d \log n$, so this algorithm will take time $\tilde{O}(n^c)$, where $c = 2 - 3/2(k+1)$. This is a constant strictly less than two. This running time is achieved only with probability $9/10$. By stopping the algorithm after running for this many steps and declaring the graph is bipartite, we get an algorithm that always runs in time $\tilde{O}(n^c)$ accepts all bipartite graphs, and rejects non-bipartite ones with probability at least $3/5$.

\section{Proof of the conjecture for regular graphs}
\label{sec:ak}

We prove that Conjecture~\ref{conj:ak} is true when the graph is $d$-regular for any $d$. The proof can easily be extended to almost-regular graphs, where the degree of every vertex is between $d$ and $Kd$ for an arbitrary constant $K$, but we won't do so for simplicity.

\begin{theorem}
\label{thm:ak}
Assume $G$ is $\eps$-far from bipartite and $d$-regular. Let $S$ be a random subset of $O((\log(1/\eps))^2/\eps)$ vertices of $G$. With probability $9/10$, $S$ is $\tilde{\Omega}(\eps)$-far from bipartite.
\end{theorem}

Let $d = k\eps n$. We will think of the set $S$ as consisting of several parts, namely $S = S_1\cup\dots\cup\/S_s\cup\/T$, where $s = 8k\log(1/\eps)$. The set $S_i$ is chosen by putting each element in it indepenently at random with probability $10\log(1/\eps)/k\eps\/n$. The set $T$ is chosen by putting each element in it independently at random with probability $p = 100\log(1/\eps)/\eps n$. 

Say the pair $(S^*, T)$ is {\em violating} if for every partition of $S^* \cup T$, either $T$ contains $\log(1/\eps)/\eps$ monochromatic edges, or there are at least $1/k\eps$ monochromatic edges between $T$ and $S^*$.

\begin{claim}
Assume $G$ is $\eps$-far from bipartite and $S^*$ dominates all but $\eps/4$ fraction vertices of $G$. Then the pair $(S^*, T)$ is violating with probability at least $1/2$ over the choice of $T$.
\end{claim}
\begin{proof}
Fix a partition of $S^*$ and consider the induced partition on those vertices in $G$ dominated by $S^*$. (That is, we color every such vertex in $G$ using the opposite color of its lexicographically smallest neighbor in $S^*$.)  We now upper bound the probability that among those vertices of $T$ dominated by $S^*$, there are fewer than $5\log(1/\eps)/\eps$ monochromatic edges with respect to the induced partition.

Since $S_i$ fails to dominate only $\eps/4$ fraction of vertices, there must be $m \geq \eps\binom{n}{2}/2$ violating edges among those vertices of $G$ dominated by $S_i$ with respect to the induced partition. The expected number of such violating edges that make it inside $T$ is then at least $p^2m$. By an application of Janson's tail inequality (see for example~\cite{ASmethod}), we find that the probability that fewer than $p^2m/2 \geq 5\log(1/\eps)/\eps$ monochromatic edges make it inside $T$ is at most $e^{-pm/4d} = e^{-4\log(1/\eps)/k\eps}$.

By a union bound, the probability that there exists a partition of $S^*$ such that the induced partition on $T$ has fewer than $5\log(1/\eps)/\eps$ violating edges is at most $1/8$. By another large deviation bound (Chernoff + union bounds), the probability that $T$ has a vertex of degree more than $4k\log(1/\eps)$ is at most $1/8$, so with probability at least $1/4$ neither of these events hold.

Now we argue that if neither event holds then $(S^*, T)$ is a violating pair. Suppose not. Then there is a partition of $S^* \cup T$ where $T$ contains fewer than $\log(1/\eps)/\eps$ monochromatic edges and there are fewer than $1/k\eps$ monochromatic edges between $T$ and $S_i$. If we now change the color of the vertices in the partition of $T$ that are inconsistent with $S_i$ we have introduced at most $(1/k\eps)\cdot(4k\log(1/\eps))$ new violating edges, by the bound on the degree. So we obtain a partition of $S_i \cup T$ with fewer than $\log(1/\eps)/\eps + 4\log(1/\eps)/\eps \leq 5\log(1/\eps)/\eps$ monochromatic edges, contradicting the first assumption.
\end{proof}

Now let $S^*$ be a set chosen from the same distribution as $S_i$. Since $S^*$ dominates all but an $\eps/4$-fraction of vertices with probability at least $3/4$, we get that
\[ \pr_{S^*, T}[\text{$(S^*, T)$ is violating}] \geq 1/2 \]
and
\[ \pr_{T}[\Pr_{S^*}[\text{$(S^*, T)$ is violating}] \geq 1/4] \geq 1/4. \]
Fix a $T$ such that the inner probability is at least $1/4$. By Chebyshev's inequality,
\[ \Pr_{S_1,\dots,S_s}[\text{at least $s/8$ pairs $(S_i, T)$ are violating}] \geq 7/8 \]
So by a union bound,
\[ \Pr_{T,S_1,\dots,S_s}[\text{at least $s/8$ pairs $(S_i, T)$ are violating}] \geq 1/8. \]
We now show that if there are $s/8$ such violating pairs, $S$ must be $\log(1/\eps)/\eps$-removed from bipartite. Consider an arbitrary partition of $S$. If $T$ contains $\log(1/\eps)/\eps$ monochromatic edges then we are done. Otherwise, There are at least $1/k\eps$ monochromatic edges between $T$ and $s/8$ of the sets $S_i$. This also gives $s/8k\eps = \log(1/\eps)/\eps$ monochromatic edges.

To amplify the success probability from $1/8$ to $9/10$ we take $20$ random disjoint copies of $S$.

\bibliographystyle{alpha}

\begin{thebibliography}{bipart-test}

\bibitem[AFKS99]{AFKS}
Noga Alon, Eldar Fischer, Michael Krivelevich, and Mario Szegedy.
\newblock Efficient testing of large graphs.
\newblock {\em FOCS '99: Proceedings of the 40th Annual IEEE Symposium on
  Foundations of Computer Science}, 0:656, 1999.

\bibitem[AFNS06]{AFNS}
Noga Alon, Eldar Fischer, Ilan Newman, and Asaf Shapira.
\newblock A combinatorial characterization of the testable graph properties:
  it's all about regularity.
\newblock In {\em STOC '06: Proceedings of the 38th Annual ACM Symposium on
  Theory of Computing}, pages 251--260, New York, NY, USA, 2006. ACM.

\bibitem[AK02]{AK}
Noga Alon and Michael Krivelevich.
\newblock Testing $k$-colorability.
\newblock {\em SIAM J. Discret. Math.}, 15(2):211--227, 2002.

\bibitem[AS05]{AS}
Noga Alon and Asaf Shapira.
\newblock Every monotone graph property is testable.
\newblock In {\em STOC '05: Proceedings of the 37th Annual ACM Symposium on
  Theory of Computing}, pages 128--137, New York, NY, USA, 2005. ACM.

\bibitem[AS08]{ASmethod}
Noga Alon and Joel~H. Spencer.
\newblock {\em The probabilistic method}.
\newblock Wiley and Sons, Inc., 3rd edition, 2008.

\bibitem[BOT02]{BOT}
Andrej Bogdanov, Kenji Obata, and Luca Trevisan.
\newblock A lower bound for testing 3-colorability in bounded-degree graphs.
\newblock In {\em FOCS '02: Proceedings of the 43rd Annual IEEE Symposium on
  Foundations of Computer Science}, pages 93--102, Los Alamitos, CA, USA, 2002.
  IEEE Computer Society.

\bibitem[BSGH{\etalchar{+}}06]{BSGH}
Eli Ben-Sasson, Oded Goldreich, Prahladh Harsha, Madhu Sudan, and Salil Vadhan.
\newblock Robust pcps of proximity, shorter pcps, and applications to coding.
\newblock {\em SIAM Journal on Computing}, 36(4):889--974, 2006.

\bibitem[BT04]{BT}
Andrej Bogdanov and Luca Trevisan.
\newblock Lower bounds for testing bipartiteness in dense graphs.
\newblock In {\em Proceedings of the 19th Annual IEEE Conference on
  Computational Complexity}, pages 75--81, Los Alamitos, CA, USA, 2004. IEEE
  Computer Society.

\bibitem[DR06]{DR}
Irit Dinur and Omer Reingold.
\newblock Assignment testers: Towards a combinatorial proof of the pcp theorem.
\newblock {\em SIAM Journal on Computing}, 36(4):975--1024, 2006.

\bibitem[FN07]{FN}
Eldar Fischer and Ilan Newman.
\newblock Testing versus estimation of graph properties.
\newblock {\em SIAM Journal on Computing}, 37(2):482--501, 2007.

\bibitem[GGR98]{GGR}
Oded Goldreich, Shafi Goldwasser, and Dana Ron.
\newblock Property testing and its connection to learning and approximation.
\newblock {\em J. ACM}, 45(4):653--750, 1998.

\bibitem[GKNR08]{GKNR}
Oded Goldreich, Michael Krivelevich, Ilan Newman, and Eyal Rozenberg.
\newblock Hierarchy theorems for property testing.
\newblock {\em Electronic Colloquium on Computational Complexity (ECCC)},
  15(097), 2008.

\bibitem[GR98]{GolR98}
Oded Goldreich and Dana Ron.
\newblock A sublinear bipartiteness tester for bounded degree graphs.
\newblock In {\em STOC '98: Proceedings of the 30th Annual ACM Symposium on
  Theory of Computing}, pages 289--298, New York, NY, USA, 1998. ACM.

\bibitem[GR08]{GR}
Mira Gonen and Dana Ron.
\newblock On the benefits of adaptivity in property testing of dense graphs.
\newblock {\em Algorithmica}, 2008.
\newblock Preliminary version in {\em Proceedings of RANDOM 2007}.

\bibitem[GR09]{GolR09}
Oded Goldreich and Dana Ron.
\newblock Algorithmic aspects of property testing in the dense graphs model.
\newblock In {\em Proceedings of APPROX / RANDOM 2009}, pages 520--533, Berlin,
  Heidelberg, 2009. Springer-Verlag.

\bibitem[GT03]{GT}
Oded Goldreich and Luca Trevisan.
\newblock Three theorems regarding testing graph properties.
\newblock {\em Random Struct. Algorithms}, 23(1):23--57, 2003.

\bibitem[KKR04]{KKR}
Tali Kaufman, Michael Krivelevich, and Dana Ron.
\newblock Tight bounds for testing bipartiteness in general graphs.
\newblock {\em SIAM J. Comput.}, 33(6):1441--1483, 2004.

\end{thebibliography}
\newcommand{\etalchar}[1]{$^{#1}$}

\end{document}